\documentclass[a4paper, 10pt, conference]{ieeeconf}      

\IEEEoverridecommandlockouts                              
\usepackage{fancyhdr}

\usepackage{graphicx}      
\usepackage{color}
\usepackage{booktabs,multirow,multicol}
\usepackage{amsmath}
\usepackage{amsfonts}
\usepackage{amssymb}
\usepackage{subfigure}
\usepackage{wrapfig}

\newtheorem{definition}{Definition}
\newtheorem{theorem}{Theorem}

\newtheorem{corollary}{Corollary}
\newtheorem{proposition}{Proposition}

\newcommand{\eg}{\emph{e.g.}~}

\newcommand{\Eg}{\emph{E.g.}~}
\newcommand{\ie}{\emph{i.e.}}

\newcommand{\hspm}{\hspace{-2pt}}
\newcommand{\hspu}{\hspace{-1pt}}

\newcommand{\bexk}{\mathcal{B}_{{\textrm{ex},k}}}

\newcommand{\bex}{\mathcal{B}_{{\textrm{ex}}}}

\newcommand{\nato}{{\mathbb{N}_0}}

\newcommand{\taut}{{_{[\tau,t]}}}

\newcommand{\akl}{\mathcal{A}_k^{-1}}
\newcommand{\ak}{\mathcal{A}_k}

\newcommand{\tautn}{{{[\tau,t]}}}

\title{\LARGE \bf
Decentralized set-valued state estimation based on\\ non-deterministic chains}

\author{N. Bajcinca, Y. Kouhi, V. Nenchev, J. Raisch 
\thanks{N. Bajcinca, Y. Kouhi, V. Nenchev and J. Raisch are with \emph{Max Planck Institute for Dynamics of Complex Technical Systems}, Sandtorstr. 1, 39106 Magdeburg, Germany, and \emph{Technische Universit\"at Berlin}, Control Systems Group, Einsteinufer 17, EN 11, 10587 Berlin, Germany. Corresponding email: \footnotesize{\emph{bajcinca@mpi-magdeburg.mpg.de.}}}
}

\begin{document}

\maketitle
\thispagestyle{empty}
\pagestyle{empty}


\begin{abstract}
A general decentralized computational framework for set-valued state estimation and prediction for the class of systems that accept a hybrid state machine representation is considered in this article. The decentralized scheme consists of a conjunction of distributed state machines that are specified by a decomposition of the external signal space.
%
While this is shown to produce, in general, outer approximations of the outcomes of the original monolithic state machine, here, specific rules for the signal space decomposition are devised by utilizing structural properties of the underyling transition relation, leading to a recovery of the exact state set results. By applying a suitable approximation algorithm, we show that computational complexity in the decentralized setting may thereby essentially reduce as compared to the centralized estimation scheme.

\end{abstract}


\section{Introduction}
\label{sec:introduction}
%
Set-valued state computation is often used in the analysis and synthesis of complex systems. As state sets are thereby guaranteed to contain the true state of the system, such a computational approach can be efficiently employed for the prediction of the system's behavior involving physical and measurement uncertainties. 
For instance, reachability analysis for verification of safety specifications is a typical application in this context. In hybrid systems, abstraction-based approaches naturally lead to a set-valued computational framework, see \eg \cite{Lunze2010}. Yet, the computational complexity remains often prohibitive, which has been an impetus for the increasing interest on the decentralized state estimation and prediction schemes, particularly in the area of discrete event systems with applications to failure detection and diagnosis, see \eg \cite{fabrecdc00, Xu2009}.



In this paper, a general framework for decentralized estimation and prediction for a class of hybrid and discrete event systems with a finite external signal space is considered. Therefore, initially the original signal space is decomposed into a finite number of aggregate signal spaces of a lower cardinality, which physically may be interpreted as substitution of a ``high resolution sensor'' by a finite set of ``coarser'' ones. Thereby, each of the introduced aggregate signal spaces invokes a distributed state machine by ``relabeling'' the symbols of the  monolithic state machine. However, due to the reduction in the measurement resolution, state set predictions of individual distributed state machines, and as a consequence, of the whole decentralized scheme itself, is indeed  conservative. To obviate this, the ``coarse sensors'' must be appropriately designed in the sense that the intersection of the individual computation outcomes resulting thereof, invariably leads to exact state set estimates. To this end, in this article, specific signal space decomposition rules are constructed by utilizing the structural properties of the transition relation of the underlying monolithic state machine. A simple algorithm is developed based on the concept of so-called ``non-deterministic chains'', which represent a special class of transition relations featuring inherent injective transition functions.
%
Our computational framework relates to the work \cite{wodes2010} which focuses on distributed set-valued state estimation for I/S/O machines. Yet it includes additional perspectives of decentralized computation. For instance, one of the primers here has been the development of the guidelines for the design of the decentralized scheme itself. Moreover, our framework covers the larger class of state machines that preserve injectivity in the state transition function such as non-deterministic automata or I/S/- representations with singleton output maps. Recently, this approach has been further generalized in \cite{bajc_rom:2011} to the class of systems involving  non-injective state transition functions.

The remainder of the paper is organized as follows. 
In Section~\ref{sec:monolithic}, we recall a few basic concepts from the behavioral system theory. In addition, a recursive algorithm for the computation of the set-valued state estimates and predictions is presented. Section~\ref{sec:distributed} represetns the core of the work. Here, we define the decentralized computational scheme, and introduce external signal space decomposition using specific aggregation functions. In particular, the concept of {non-deterministic chains} and the corresponding main algorithm for the signal space decomposition are presented. In Section~\ref{sec:Application}, the results are employed in a decentralized estimation setup using $\ell$-complete approximations. In addition, a comparison of the space/time complexity between the decentralized and monolithic scheme is shortly discussed. Finally, an illustrative numerical example for decentralized set-valued state estimation is included.

We use the following notation: capital letters denote signal spaces, \eg $X$, $U$, $Y$, $W$ represent the state space,
the input and output spaces, and the external signal space, respectively. The corresponding elements are denoted by greek lowercase letters, \eg $W=\{\omega_1,\ldots,\omega_m\}$. We consider the discrete time domain, hence signals, which
are denoted by lowercase letters, are sequences of symbols from the appropriate signal space, \eg $w : \nato\to W$ represents the external signal. The restriction of a signal to an interval $[\tau,t]$, $\tau,t\in\nato$, $0\leq \tau\leq t$, is denoted by $\cdot|\taut$, \eg $w|\taut=  w(\tau)\ldots w(t)$. The space of the finite sequences (strings) $w|\taut$ will be denoted as $W^\tautn= W^{t-\tau+1}$. The string $w|\taut$ will be considered as an element of $W^\tautn$, i.e. they will be represented by an $(t\hspu-\hspu\tau\hspu+\hspu 1)$-tu\-ple ordered by the time parameter. Finally, let $f$ be an arbitrary function $f$ defined on some domain $X$. If $\Xi$ is a subset of $X$, we will use the convention $f(\Xi):= \cup_{x\in \Xi} f(x)$. For singleton sets we avoid usage of brackets. 

\section{Preliminaries}
\label{sec:monolithic}

\subsection{Systems \& realizations}
\label{sec:monolithic_dynamic}
This section provides basic system and realization concepts from the behavioral perspective, see \eg \cite{Wi91}. A dynamical system $\Sigma$ is defined as a triple $(T,W,\mathcal{B})$, with \emph{time axis} $T\subseteq \mathbb{R}$, the \emph{external signal space} $W$, and the \emph{behavior} $\mathcal{B}\subseteq W^T$, where $W^T=\{w: T\to W\}$.
Throughout this article, the discussion is confined to discrete-time systems, hence $T=\mathbb{N}_0$. The external signal space is finite: $W=\{\omega_1,\ldots,\omega_m\}$. $\mathcal{B}$ then represents a set of sequences $w:\mathbb{N}_0\rightarrow W$ which are compatible with the dynamics of the system $\Sigma$. 
%
%
%
A dynamical system $\Sigma=(\mathbb{N}_0,W,\mathcal{B})$ is said to be {time invariant} if $\sigma\mathcal{B}\subseteq \mathcal{B}$, where $\sigma$ is the backwards time shift operator: $\sigma w(t):=w({t+1})$, $t \in \mathbb{N}_0$, and $\sigma\mathcal{B}:=\{\sigma w; w\in \mathcal{B}\}$; for $\tau>1$, $\sigma^\tau:=\sigma \sigma^{\tau-1}$ . 

A state machine is a tuple $P=(X,W,\Delta,X_0)$ where $X$ denotes the state space, $W$ the external signal space, $\Delta\subseteq X\times W\times X$ the transition relation, and $X_0\subseteq X$ the initial state set. If $X=\mathbb{R}^n\times D$, where $n\in\mathbb{N}_0$ and $|D|\in\mathbb{N}<\infty$, $P$ is a hybrid state machine; for $n=0$, $P$ is a finite state machine. Throughout the paper, $P$ is assumed to be non-blocking, that is for all $\xi\in X$ there exists a $\omega\in W$ such that $(\xi,\omega,\xi')\in \Delta$. Furthermore, we assume $X_0=X$. Then, $\sigma\mathcal{B}=\mathcal{B}$, implying that $\Sigma$ is time-invariant.

For systems exhibiting an input/output structure,  the external signal space $W$ can be decomposed as $W=U\times Y$, with $U$ and $Y$ being the sets of input and output symbols. Then, $P=(X,U\times Y,\Delta,X_0)$, is said to be an {I/S/- machine} if for each state $\xi\in X$ and each $\mu\in U$, there exists a $\nu\in Y$ and a $\xi^\prime\in X$, such that $(\xi,(\mu,\nu),\xi^\prime)\in\Delta$. If $\nu$ and $\xi^\prime$ are unique for all $\xi\in X$ and $\mu\in U$, $P$ is said to be an {I/S/O machine}. Note that I/S/O machines are deterministic by definition. A state machine $P=(X,W,\Delta,X_0)$ induces a state space system $\Sigma_\text{S}=(\mathbb{N}_0,W\times X,\mathcal{B}_\text{S})$, where $\mathcal{B}_\text{S}$ is referred to as the {full behavior}, and is defined as
\begin{align}
\mathcal{B}_\text{S}\hspm:=\hspm\{(w,\hspm x); (x(t),\hspm w(t),\hspm x({t\hspm+\hspm1}))\hspm\in\hspm \Delta, t\hspm\in\hspm \mathbb{N}_0,x_0\hspm\in\hspm X_0\}.
\end{align}
The external behavior $\mathcal{B}_{\text{ex}}$ of $\Sigma_\text{S}$ is then defined to be the projection of $\mathcal{B}_\text{S}$ onto $W^{\mathbb{N}_0}$, that is $\mathcal{B}_{\text{ex}}:=\mathcal{P}_W\mathcal{B}_\text{S}=\{w; \exists x\in X^{\mathbb{N}_0}, (w,x) \in \mathcal{B_\text{S}}\}$. Finally, a state machine $P=(X,W,\Delta,X_0)$ with induced external behavior $\mathcal{B}_{\text{ex}}$ is said to be a realization of a dynamical system $\Sigma=(\mathbb{N}_0,W,\mathcal{B})$ if $\mathcal{B}_{\text{ex}} = \mathcal{B}$. This will be denoted by  $P \cong \Sigma$.

\subsection{State set estimation \& prediction}
\label{sec:monolithic_setvaluedset}
Let $\mathcal{B}_\text{S}$ and $\mathcal{B}_\textrm{ex}$ be the induced full and external behavior of the state machine $P=(X,W,\Delta,X)$, respectively. Define state sets compatible to an external string $w|_{[\tau,t]}\in W^{\tautn}$ at the time instants $t$ and $t+1$ as
\begin{subequations}
\begin{align}
\label{eq:chi_defs}
\hspm\hspm\hspm\hspm\hspm\chi(\hspu w|_{[\tau,t]}\hspu)\hspm:=&\hspu\{\xi;\exists (w^\prime\hspm\hspm,x)\hspm\in\hspm\mathcal{B}_\text{S},\hspu x(t)\hspm=\hspm\xi,\hspu w^\prime|_{[\tau,t]}\hspm=\hspm w|_{[\tau,t]}\hspu\}, \\
\label{eq:rho_defs}
\hspm\hspm\hspm\hspm\hspm\rho(\hspu w|_{[\tau,t]}\hspu)\hspm:=&\hspu\{\xi;\exists (w^\prime\hspm\hspm,x)\hspm\in\hspm\mathcal{B}_\text{S},\hspu x({t\hspm+\hspu\hspm1})\hspm=\hspm\xi,\hspu w^\prime|_{[\tau,t]}\hspm=\hspm w|_{[\tau,t]}\hspu\}.
\end{align}
\end{subequations}
Both, $\chi$ and $\rho$, are families of set-valued functions $W^{\tautn}\to 2^X$ parametrized by a restriction interval $[\tau,t]$. Note that $w|\taut\in\bex|\taut\Leftrightarrow \chi(w|\taut)\neq\emptyset$. In general, more information about the past leads to more accurate state estimates. This fact is reflected by the following proposition.
\begin{proposition}
\label{thm:main_1}
Consider a machine $P=(X,W,\Delta,X)$ with induced external behavior $\mathcal{B}_{\text{ex}}$, and let $w\in \mathcal{B}_\textrm{ex}$. Then
\vspace{-4pt}
\begin{subequations}
\begin{align}
\label{eq:chis_theorem}
\chi(w|_{[0,t]}) \subseteq \chi(w|_{[1,t]})\subseteq \ldots \subseteq \chi(w|_{[t,t]}),\\
\label{eq:rhos_theorem}
\rho(w|_{[0,t]})\subseteq \rho(w|_{[1,t]})\subseteq \ldots \subseteq \rho(w|_{[t,t]}).
\end{align}
\end{subequations}
\end{proposition}
\begin{proof}
Introduce the set of behaviors $\mathcal{B}_{\tau} := \{w^\prime \in \mathcal{B}_{\text{ex}}; {w}^\prime|_{[\tau,t]}=w|_{[\tau,t]}\}$, which contains all sequences in $\mathcal{B}_\textrm{ex}$ that share the same restriction $w|_{[\tau,t]}$. Then, by definition (\ref{eq:chi_defs}), $\chi(w|_{[\tau,t]})=\chi(\mathcal{B}_{\tau}):=\cup_{w^\prime\in\mathcal{B}_{\tau}}\chi(w^\prime|_{[\tau,t]})$. It is obvious that $
\mathcal{B}_{0} \subseteq \mathcal{B}_{1} \subseteq \ldots \subseteq \mathcal{B}_{t}$, which implies $\chi(\mathcal{B}_{0}) \subseteq \chi(\mathcal{B}_{1}) \subseteq \ldots \subseteq \chi(\mathcal{B}_{t})$. This proves (\ref{eq:chis_theorem}). Equation (\ref{eq:rhos_theorem}) follow by same lines of argument.
\end{proof}

We introduce now few one-step prediction expressions, which will reveal an iterative computation procedure for the state sets in \eqref{eq:chi_defs} and \eqref{eq:rho_defs}. Therefore, introduce the parametrized {state transition} function $\hat\rho_\omega: X\to 2^X$ as
\begin{subequations}
\begin{align}
\label{eq:rho_hat}
\hat\rho_\omega(\xi) &:=  \{\xi^\prime ; (\xi,\omega,\xi^\prime)\hspm\in\hspm\Delta\}.
\end{align}
For $\Omega\subseteq W$ and $\Xi\subseteq X$ in accordance with the adopted notation convention we define
\begin{align}
\label{eq:rho_hat}
\hat\rho_\Omega(\Xi) :=  \cup_{\omega\in\Omega, \xi\in\Xi}~\hat\rho_\omega(\xi).
\end{align}
The predicted states in $X$, resulting from the occurrence of the symbol $\omega\in W$, are then computed by
\begin{align}
\label{eq:rho}
{\rho}(\omega)= \hat{\rho}_\omega({\chi(\omega)})
\end{align}
where  according to (\ref{eq:chi_defs})
\begin{align}
\label{eq:chi}
{\chi}(\omega) = \{\xi; \exists \xi^\prime, (\xi,\omega,\xi^\prime)\in\Delta\},
\end{align}
with $\omega=w|_{[t,t]}$. For sequences, we get analogously
\begin{align}
\label{eq:rho_recursive}
\rho(w|\taut) = \hat\rho_{w(t)}(\chi(w|\taut)).
\end{align}
Moreover, by definition \eqref{eq:chi_defs}
\begin{align}
\label{eq:chi_recursive}
\chi(w|_{[\tau,t]})=   \rho(w|_{[\tau,t-1]})\cap \chi(w({t})),
\end{align}
\end{subequations}
which along with \eqref{eq:rho_recursive} reveals a recursive structure in computing $\chi(w|\taut)$ and $\rho(w|\taut)$.


\section{Decentralized computation scheme}
\label{sec:distributed}

\subsection{Signal space decomposition}
\label{sec:abstraction}
Consider the external signal space $W$, as defined in Section~\ref{sec:monolithic}, and introduce a finite set of \emph{aggregation functions}
\begin{equation}
\label{eq:aggs}
\mathcal{A}_k: W\to V_k,~~k\in\{1,2,\ldots,p\}
\end{equation}
where $|V_k|\leq |W|$. The functions (\ref{eq:aggs}) are required to fulfill the following \emph{resolvability} or \emph{consistency} condition
\begin{align}
\label{eq:consistency}
\cap_{k=1}^{p} \mathcal{A}_k^{-1}(\mathcal{A}_k(\omega))=\omega~~~(\omega\in W)
\end{align}
where the inverse mapping $\mathcal{A}^{-1}_k:V_k\to 2^W$, is defined as
\begin{align}
\mathcal{A}_k^{-1}({\theta_k}):=\{\omega\in W; \mathcal{A}_k(\omega)={\theta_k} \}.
\end{align}
Due to consistency, each symbol $\omega \in W$ is uniquely resolved by an ordered $p$-tuple $({\theta_1},\ldots, {\theta_p})$, where ${\theta_k}=\mathcal{A}_k(\omega)$, $k\in\{1,\ldots,p\}$. We refer to this as a \emph{decomposition} of the original signal space $W$ into $p$ signal spaces, and write
\begin{align}\label{eq:decomposition}
W \leadsto V_1\times V_2\times \ldots \times V_p.
\end{align}
In general, the opposite does not hold: not every $p$-tuple in $V_1\times V_2\times \ldots \times V_p$ will be associated with a symbol in $W$.

Now extend the definition of aggregation functions to $\ak: W^\tautn \to V_k^\tautn$, by a symbolwise mapping. That is,  $v_k|\taut=\ak(w|\taut)$ if $v_k(l)=\ak(w(l))~(\tau\leq l \leq t)$. Then, it is clear that the resolvability condition \eqref{eq:consistency} carries over to strings, as well, that is
\begin{align}
\label{eq:consistency_seq}
\cap_{k=1}^{p} \mathcal{A}_k^{-1}(\ak(w|\taut))=w|\taut ~~(w\in W^\nato).
\end{align}

%

\subsection{Distributed state machines}
\label{sec:distributed_estimation}
Having introduced the signal spaces $V_k$, $k\in\{1,\dots,p\}$, each is now associated with a \emph{distributed state machine} $P_k=(X,V_k,\Delta_k,X)$, where $\Delta_k\hspm\subseteq \hspm X\hspm\times\hspm V_k \hspm\times\hspm X$ is defined as
\begin{equation}\label{eq:delta_k}
\Delta_k =\{(\xi,{\theta_k},\xi^\prime) ; \exists \omega\hspm\in\hspm \mathcal{A}_k^{-1}({\theta_k}),
(\xi,\omega,\xi^\prime)\hspm\in\hspm\Delta\}.
\end{equation}
The original state machine $P=(X,W,\Delta,X)$ is referred to as the \emph{monolithic} state machine. From (\ref{eq:delta_k}) it follows for the full and external behavior of a machine $P_k$: $\mathcal{B}_{\text{s},k}:=\{(v_k,x); v_k=\ak(w), (w,x)\in\mathcal{B}_{\text{s}}\}$ and $\mathcal{B}_{\text{ex},k}=\ak(\bex):=\{\ak(w); w\in\bex\}$, respectively. The definitions for the estimation and prediction functions $\chi_{k}: V_k^{\tautn} \to 2^X$ and $\rho_k:V_k^{\tautn} \to 2^X$, $k\in\{1,\dots,p\}$, are analogous to those in (\ref{eq:chi_defs}-\ref{eq:rho_defs}) for the monolithic machine $P$. Equivalently stated:
\begin{align}
\label{eq:chirho_maps}
\rho_k=\rho\circ\akl,~~\chi_k=\chi\circ\akl.
\end{align}
For instance,
\begin{align}
\label{eq:chis_k}
\rho_k(v_k|\taut) = \rho(\akl(v_k|\taut)) =\cup_{l=1}^{s_k}\rho(w_{kl}|\taut),
\end{align}
for some $s_k\in \mathbb{N}$. 
The external behavior $\bexk$ of the machine $P_k$ is a coarse approximation of the ``monolithic'' behavior $\bex$, in that $w|\taut\in \akl(v_k|\taut)$ if $v_k=\ak(w)$, for any $w\in\bex$. However, the resolvability condition $\cap_{k=1}^p\akl(v_k|\taut) = w|\taut$ suggests a computation scheme with $p$ distributed state machines $P_k$ including an intersection of the respective outcomes. Therefore, consider a string $w|\taut$ corresponding to the $p$-tuple $(v_1|\taut,\dots, v_p|\taut)$. Then, in general, it follows
\begin{subequations}
\begin{align}
\cap_{k=1}^{p}\chi_k(v_k|_{[\tau,t]})&= \cap_{k=1}^{p} \chi(\akl(v_k|\taut))\notag \\
{} & \hspace{.1cm}{\supseteq \chi(\cap_{k=1}^{p}\akl(v_k|\taut))}\notag \\
{} & \hspace{.1cm}{=\chi(w|\taut)},
\end{align}
where we use the fact: $\phi(M_1)\cap \phi(M_2)\supseteq \phi(M_1\cap M_2)$, and the consistency condition. Similarly,
\begin{align}
\label{eq:thm2:2}
\cap_{k=1}^{p}\rho_k(v_k|\taut)& \supseteq \rho(w|\taut).
\end{align}
\end{subequations}
As a consequence, the parallel computation scheme provides, in general, an overapproximation  of the outcomes of the original monolithic state machine. However, for certain classes of transition relations $\Delta$, specific consistent functions $\ak, ~k\in\{1,\dots,p\}$ can be constructed, which lead to exact computation results in the decentralized scheme. The basic concept which we use in constructing such aggregation functions is that of ``non-deterministic chains''.

\subsection{Non-deterministic chains}
\label{sec:Chain-structured state machines}
\begin{definition}\label{def:chain1}
Consider a state machine $P\hspm=\hspm(X,W,\Delta,X)$, and let $\Omega\subseteq W$. A transition subrelation $\delta\subseteq \Delta$ is said to be a \emph{non-deterministic chain} over $\Omega$ if for all $\xi\in X$ and $\omega', \omega'' \in \Omega$:
\vspace{-6pt}
\begin{align*}
\text(i)~ & (\xi,\omega',\xi')\in\delta, (\xi,\omega'',\xi'')\in\delta\Rightarrow \omega'=\omega'',\\
\text(ii)~ & (\xi',\omega',\xi)\in\delta, (\xi'',\omega'',\xi)\in\delta\Rightarrow \xi'=\xi''.
\end{align*}
\end{definition}
The transition relation $\delta$ can naturally be associated with the functions $\chi^c : \Omega\to 2^ X$, $\rho^c : \Omega\to 2^X$, and the transition function $\hat\rho^c: \chi(\Omega)\to 2^X$, defined as
\begin{align}\label{eq:ndc_funs}
\chi^c :=\chi|{_{\Omega}},~~\rho^c :=\rho|{_{\Omega}},~~\hat\rho^c:=\hat\rho_\Omega.
\end{align}
Then, \emph{(i)} can be equivalently restated as $\chi^c(\omega')\cap\chi^c(\omega'')\neq\emptyset\Rightarrow \omega'=\omega''$,  while \emph{(ii)} is equivalent to $\hat\rho^c(\xi')\cap\hat\rho^c(\xi'')\neq\emptyset\Rightarrow \xi'_1=\xi''$, leading to the following statement.
\begin{proposition}\label{def:chain2}
A transition relation $\delta\subseteq \Delta$ is a non-deterministic chain if and only if $\chi^c$ and $\hat\rho^c$ are absolutely injective set-valued maps.
\end{proposition}
In the sequel, we introduce a systematic method for signal space decomposition which results from partitioning the transition relation $\Delta$ into a finite set of non-deterministic chains. To this end, for a given monolithic machine $P=(X,W,\Delta,X)$, suppose that partitionings
\begin{align}\label{eq:decomW}
W={\cup}_{j=1}^{r} \Omega_j,~~\text{and}~~\Delta = {\cup}_{j=1}^{r} \delta_j
\end{align}
exist, such that each transition subrelation $\delta_j\subseteq \Delta$ over $\Omega_j$ represents a non-deterministic chain. Then, we claim that the state machine $P$ is \emph{chain-decomposable}. The following example illustrates this idea.

\emph{Example 1:}
Consider an I/S/- machine $P=(X,U\times Y,\Delta,X)$ with singleton output maps. Let $U=\{\mu_j; j=1,\dots,r\}$, and introduce the partitioning $W=U\times Y=\cup_{j=1}^r \Omega_j$, where $\Omega_j := \{\mu_j\}\times Y$. This induces a partitioning of the transition relation $\Delta=\cup_{j=1}^r  \delta_j$. By definition, functions $f\hspm :\hspm X\hspm \times\hspm  U\hspm \to\hspm  2^X$ and $h\hspm :\hspm X\hspm \times\hspm  U\hspm \to\hspm  Y$ exist, such that
\begin{align}\label{eq:isminus}
\hspm(\xi,(\mu_j,\nu),\xi^\prime)\in\delta_j \Leftrightarrow \xi'\hspm\in \hspm f(\xi,\mu_j),~\nu\hspm =\hspm h(\xi,\mu_j).
\end{align}
Then, each state $\xi\in X$ can be associated with a unique symbol pair $(\mu_j,\nu)\in \Omega_j$. Hence, $(i)$ in Definition~\ref{def:chain1} is fulfilled. Define further $\hat\rho^c_j:X \to 2^X$ as $\hat\rho^c_j(\xi):=f(\mu_j,\xi)$, and let it be absolutely injective. Then, $\delta_j$ is a non-deterministic chain for all $j\in\{1,\dots,r\}$.~\QED

Next, introduce a consistent signal space decomposition as discussed in Section~\ref{sec:abstraction}, for $\Omega_j$, $j\in\{1,\dots,r\}$
\begin{align}
\label{eq:decomp_chains}
\Omega_j \leadsto V_{j,1}\times\dots\times V_{j,p}.
\end{align}
Notice that such decompositions invariably yield a consistent decomposition  for the original signal space $W\leadsto V_1\times\dots\times V_p$ if, \eg
\begin{align}\label{eq:Qk_chains}
V_k = {\cup}_{j=1}^{r} V_{j,k},
\end{align}
The state machines $P_k$, $k\in\{1,\ldots,p\}$ are then easily constructed using the procedure described in Section~\ref{sec:distributed_estimation}. We now want to show that the intersection of their estimates and predictions produces the exact outcomes of the underlying monolithic state machine $P$.

Referring to \eqref{eq:ndc_funs} and Proposition~\ref{def:chain2}, it is important to keep in mind, that the absolute injectivity of
\begin{align}\label{eq:ndc_funs_j}
\chi^c_{\hspu j}:=\chi|{_{\Omega_{\hspu j}}},~~\rho^c_{\hspu j}:=\rho|{_{\Omega_{\hspu j}}},~~\hat{\rho}^c_{\hspu j}:=\hat\rho_{\Omega_{\hspu j}},
\end{align}
is, per construction, preserved in all disjoint subspaces $\Omega_j$. Note also that for any $t\in\nato$, all elements of the inverse mapping of a symbol $\theta_k\in V_k$ belong to the same subspace $\Omega_j$ for each $k\in\{1,\ldots,p\}$ and some $j\in\{1,\ldots,r\}$. Now, fix a string $w|\taut$ and consider the corresponding tuple $(v_1|\taut,\dots, v_p|\taut)$ in the decentralized scheme. Then,
\begin{align*}
\cap_{k=1}^{p}\rho_k(v_k|\taut)&= \cap_{k=1}^{p}\cup_{l_k=1}^{s_k}\rho(w_{kl_k}|\taut)\\
{} & \hspace{-2cm}{= \cup_{l_1=1}^{s_1}\cdots \cup_{l_p=1}^{s_p} \cap_{k=1}^p\rho(w_{kl_k}|_{[\tau,t]})}\\
{} & \hspace{-2cm}{= \cup_{l_1=1}^{\bar{s}_1}\cdots \cup_{l_p=1}^{\bar{s}_p} \cap_{k=1}^p\rho( w_{kl_k}|_{[\tau,t-1]}w(t))}
\end{align*}
where $\bar{s}_k\leq s_k$ is the number of all strings in $W^{\taut}$ of the form $w_{kl_k}|_{[\tau,t-1]}w(t)$, $k\in\{1,\dots,p\}$. In the $3^\text{rd}$ line we took advantage of the consistency condition and the fact that $\rho^c_j: \Omega_j \to 2^{X}$ is absolutely injective. All the strings $\akl(v_k|\taut)$ that do not end with $w(t)$ are then neglectable. Moreover, using the recursive formula \eqref{eq:chi_recursive}, and the injectivity of the transition function $\hat\rho^c_j$,
\begin{align*}
\cap_{k=1}^{p}\rho_k(v_k|\taut)&=  \\
{} & \hspace{-2cm}{= \cup_{l_1=1}^{\bar{s}_1}\cdots \cup_{l_p=1}^{\bar{s}_p} \cap_{k=1}^p\hat{\rho}^c_{j}\left[{\chi}^c_j(w(t))\cap\rho(w_{kl_k}|_{[\tau,t-1]})\right]}\\
{} & \hspace{-2cm}{= \hat{\rho}^c_{j}\left[{\chi}^c_j(w(t))\cap \left(\cup_{l_1=1}^{\bar{s}_1}\cdots \cup_{{l}_p=1}^{\bar{s}_p} \cap_{k=1}^p \rho(w_{kl_k}|_{[\tau,t-1]})\right)\right]}\\
{} & \hspace{-2cm}{= \hat{\rho}^c_{j}\left[{\chi}^c_j(w(t))\cap \left(\cap_{k=1}^{p} \rho_k(v_k|_{[\tau,t-1]})\right)\right]}.
\end{align*}
As a consequence, a recursive expression is obtained, which transfers the computation task from the interval $[\tau,t]$ to $[\tau,t-1]$. Repeating the recursion until $[\tau,\tau]$, leads to an expression on the right-hand side equal to $\rho(w|\taut)$, which is the proof of the following main statement.
\begin{theorem}\label{thm:main_2}
Let $P\hspm\hspm=\hspm\hspm(X,W,\Delta,X)$ be chain decomposable. Then, a decentralized scheme including the state machines $P_k$, $k\in\{1,\dots,p\}$ induced by (\ref{eq:Qk_chains}) provides exact state estimates and predictions, that is
\vspace{-2mm}
\begin{subequations}
\begin{align}
\label{eq:thm3:2}
\cap_{k=1}^{p}\chi_k(v_k|_{[\tau,t]})& = \chi(w|_{[\tau,t]}),\\
\label{eq:thm3:1}
\cap_{k=1}^{p}\rho_k(v_k|_{[\tau,t]})& = \rho(w|_{[\tau,t]}).
\end{align}
\end{subequations}
\end{theorem}
\vspace{2mm}
Following the discussion in \emph{Example~1} for I/S/- state machines, we further conclude.
\begin{corollary}[I/S/-]\label{cor:main_iso}
Consider the class of I/S/- realizations $P=(X,U\times Y,\Delta,X)$, which fulfills (\ref{eq:isminus}). Introduce a partitioning of $\Delta$ as described in \emph{Example~1}. Then, a decentralized scheme built upon any consistent decomposition with $V_k = U\times \ak(Y)$, $k\in\{1,\dots,p\}$, provides exact computation results of the form (\ref{eq:thm3:2}) and (\ref{eq:thm3:1}).
\end{corollary}

\emph{Example 2:}
In order to illustrate the proposed method, consider the automaton in Fig.\,\ref{fig:des}. Its external signal space can be
\begin{wrapfigure}{l}{30mm}
\vspace{-15pt}
\begin{center}
\includegraphics[width=29mm]{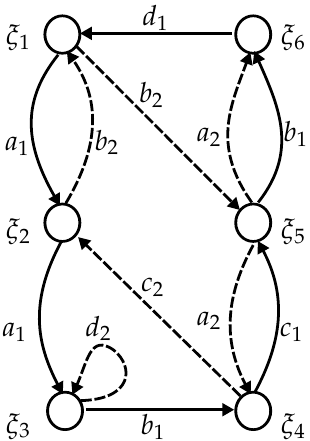}
\vspace{-17pt}
\caption{\hspm\hspm\hspm Finite\,state\,machine.}
\label{fig:des}
\vspace{-10pt}
\end{center}
\end{wrapfigure}
partitioned as $W=\Omega_1\cup\Omega_2$ with $\Omega_1=\{a_1,b_1,c_1,d_1\}$ and $\Omega_2=\{a_2,b_2,c_2,d_2\}$. As indicated by the solid and dashed lines, the corresponding transition relations  $\delta_1$ and $\delta_2$ are both non-deterministic chains. 
According to the previous elaborations, any consistent set of aggregation functions can be applied on $\Omega_1$ and $\Omega_2$. \Eg a particular signal space decomposition results from
%
%
\begin{align*}
V_{1,1} =\{\theta^1_1,\theta^2_1\} & \text{ with } \theta^1_1 \gets \{a_1,b_1\}, \theta^2_1 \gets \{c_1,d_1\}, \\
V_{1,2} =\{\theta^1_2,\theta^2_2\} & \text{ with } \theta^1_2 \gets \{a_1,c_1\}, \theta^2_2 \gets \{b_1,d_1\}, \\
V_{2,1} =\{\theta^3_1,\theta^4_1\} & \text{ with } \theta^3_1 \gets \{a_2,b_2\}, \theta^4_1 \gets \{c_2,d_2\}, \\
V_{2,2} =\{\theta^3_2,\theta^4_2\} & \text{ with } \theta^3_2 \gets \{a_2,c_2\}, \theta^4_2 \gets \{b_2,d_2\}.
\end{align*}
%
%
The corresponding decomposition reads $W\leadsto V_1\times V_2$, where $V_1=V_{1,1}\cup V_{2,1}$ and $V_2=V_{1,2}\cup V_{2,2}$, leading to the distributed machines $P_1=(X,V_1,\Delta_1,X)$ and $P_2=(X,V_2,\Delta_2,X)$. Now assume that the original system accepts a string, say$w|\taut=  a_1 b_2 $. The corresponding estimate of the monolithic machine is $\chi(a_1 b_2)=\xi_2$. The distributed machines measure accordingly the strings $v_1|_{[0,1]}=\theta^1_1 \theta^3_1 $ and $v_2|_{[0,1]}=\theta^1_2 \theta^4_2$, providing the estimates $\chi_1(\theta^1_1 \theta^3_1)=\xi_2$ and $\chi_2(\theta^1_2 \theta^4_2)=\{\xi_2,\xi_3\}$, respectively. The decentralized estimate is thus given by $\chi(\theta^1_1 \theta^3_1 )\cap\chi(\theta^1_2 \theta^4_2)=\xi_2$, which is exactly the same outcome obtained by the monolithic state machine $P=(X,W,\Delta,X)$. According to Theorem\,\ref{thm:main_2}, this must hold for all strings accepted by  $P$.~\QED

\section{Decentralized estimation using \\$\ell$-complete approximation}
\label{sec:Application}
This section addresses the set-valued state estimation for time-invariant systems $\Sigma=(\mathbb{N}_0,W,\mathcal{B})$ using the $\ell$-complete approximation algorithm. Note that $\Sigma$ is said to be $\ell$-complete, \cite{MoRaOY02}, if
\vspace{-4pt}
\begin{align}\label{eq:lcomp}
w\in \mathcal{B} \Leftrightarrow \sigma^t w|_{[t,t+\ell]}\in \mathcal{B}|_{[0,\ell]},~~(t\in \mathbb{N}_0).
\end{align}
%

\subsection{$\ell$-Complete approximation}
\label{sec:lComplete approximation}
Consider a time-invariant system $\Sigma$. The model $\Sigma_\ell=(\mathbb{N}_0,W,\mathcal{B}_\ell$), $\ell\in \mathbb{N}$, is a strongest $\ell$-complete approximation of $\Sigma$ if \emph{(i)} $\mathcal{B}_\ell$ is $\ell$-complete; \emph{(ii)} $\mathcal{B}_\ell\supseteq \mathcal{B}$, and \emph{(iii)} $\mathcal{B}_\ell^\prime \supseteq \mathcal{B}$, $\mathcal{B}_\ell^\prime$ being $\ell$-complete $\Rightarrow \mathcal{B}_\ell^\prime \supseteq \mathcal{B}_\ell$.
We consider here the realization algorithm for the strongest $\ell$-complete abstraction $\Sigma_\ell$ as follows (see also \cite{wodes2010}).
\begin{definition} \label{def:lcomplete} 
The state machine $P_\ell=(Z_\ell,W,\Delta_\ell,Z_0)$ is a realization of $\Sigma_\ell$ with
\vspace{-5pt}
\begin{align*}
\text(a)~ & Z_0= W; \nonumber\\
\text(b)~ & Z_\ell:=\cup_{r=1}^\ell W^r,~\text{where}~ W^r=w^{\alpha_1}w^{\alpha_2}\dots w^{\alpha_r}; \nonumber\\
\text(c)~ & \text{transition}~\Delta_\ell:=\cup_{r=0}^\ell\Delta_\ell^r\subseteq Z_\ell\times W\times Z_\ell,~\text{defined by:}\nonumber \\
{} & \Delta_\ell^r\hspm :=\hspm\hspm  \{(w|_{[0,r\hspu-\hspu 1]}\hspm, w(r), w|_{[0,r]})\hspm: \hspm w|_{[0,r]}\hspm \in\hspm \mathcal{B}_\ell|_{[0,r]}, 1\hspm \leq\hspm  r\hspm  <\hspm \ell\},\\
{} & \Delta_\ell^\ell\hspm:=\hspm \hspm \{(w|_{[0,\ell\hspu-\hspu 1]},w(\ell), w|_{[1,\ell]})\hspm:  \hspm w|_{[0,\ell]}\hspm \in\hspm \mathcal{B}_\ell|_{[0,\ell]}\}.
\end{align*}
\end{definition}

The $\ell$-complete representation keeps track of the system's past trajectories with a sliding time window of length $\ell$. 
The resulting behavior $\mathcal{B}_\ell$ includes the sequences compatible with  $\Sigma$, while $\mathcal{B}_0\supseteq \mathcal{B}_1\supseteq \cdots \supseteq\mathcal{B}$ holds. Using Definition~\ref{def:lcomplete}, the state $\zeta \in Z_\ell$ of $P_\ell$ at a time instant $t$ reads
\begin{align}\label{eq:zeta_mon}
   \zeta(t)=\left\{\begin{array}{ll}w|_{[0,t]} &\text{if}~\, 0\leq t< \ell,\\ w|_{[t-\ell+1,t]} &\text{if}~\, t\geq \ell.\\   \end{array} \right.
\end{align}
For $t< \ell$, the state $\zeta$ is determined by the whole signal string, whereas for $t\geq\ell$ only by its suffix of length $\ell$. Therefore, it is evident that each state $\zeta(t)$ can be associated with unique estimates $\chi(\zeta(t))\subseteq X$, defined as
\begin{align}
   \chi(\zeta(t)):=\left\{\begin{array}{ll} \chi(w|_{[0,t]}) &\text{if}~\, 0\leq t< \ell,\\ \chi(w|_{[t-\ell+1,t]}) &\text{if}~\, t\geq \ell.\\   \end{array} \right.
\end{align}
As a consequence of Proposition~\ref{thm:main_1}, in the latter case, due to $\chi(w|_{[0,t]})\subseteq \chi(w|_{[t-\ell-1,t]})$,
the state sets attached to $\zeta$, provide a coarse -- yet instantaneous -- estimation of the system realization $P=(X,W,\Delta,X)$ compatible with the measurement $w|_{[0,t]}$ for $t\geq \ell$. Another consequence of Proposition~\ref{thm:main_1} suggests that the estimation accuracy can be improved by increasing $\ell$. Indeed, given two approximation automata $P_{\ell_1}$ and $P_{\ell_2}$  with $\ell_2\geq \ell_1$, the estimates for $t \geq \ell_2$ correlate as $\chi(w|_{[t-\ell_2+1,t]}) \subseteq \chi(w|_{[t-\ell_1+1,t]})$.
Clearly, the estimation accuracy is improved, as $P_{\ell_2}$ stores a larger content of information (\ie\, a longer suffix) on $w|_{[0,t]}$. A major drawback is, however, that the number of states and transitions of the $\ell$-complete automaton increases exponentially with increasing  $\ell$, which provides a motivation for the decentralized approach.

For the distributed state machines $P_k\hspm=\hspm(X,V_k,\Delta_k,X)$, $k=1,\dots,p$, as introduced in Section~\ref{sec:distributed_estimation}, the $\ell$-complete automata $P_{k,\ell}=(Z_{k,\ell},V_k,\Delta_{k,\ell},Z_{0,k})$ are obtained from the algorithm in Definition~\ref{def:lcomplete}. Then, $Z_{k,\ell}$ includes states
\begin{align}\label{eq:zeta_dec}
   \zeta_k(t):=\left\{\begin{array}{ll}  v_k|_{[0,t]} &\text{if}~\, 0 \leq t <\ell,\\ v_k|_{[t-\ell,t]} &\text{if}~\, t\geq \ell,\\   \end{array} \right.
\end{align}
which store the estimates
\begin{align}
   \chi_k(\zeta_k(t)):=\left\{\begin{array}{ll} \chi_k(v_k|_{[0,t]}) &\text{if}~\, 0\leq t <\ell,\\ \chi_k(v_k|_{[t-\ell,t]}) &\text{if}~\, t\geq \ell.\\   \end{array} \right.
\end{align}
For a given string $w|\taut = (v_1|\taut,\dots, v_p|\taut)$, the monolithic approximation $P_\ell$ and each $P_{k,\ell}$ assume unique states $\zeta\in Z_{\ell}$ and $\zeta_k\in Z_{k,\ell}$, as defined above. Then, in accordance with the elaboration in Section~\ref{sec:distributed}:
\begin{displaymath}
\cap_{k=1}^{p}\chi_k(\zeta_k(t)) \supseteq \chi(\zeta(t))~\text{or}~\cap_{k=1}^{p}\chi_k(\zeta_k(t)) = \chi(\zeta(t)).
\end{displaymath}


\subsection{Space/time complexity}
\label{sec:complexity}
Following the discussion in the last section, it is obvious that the complexity of the monolithic and distributed estimators is directly related to the number of states in the corresponding implementation with $\ell$-complete automata. In accordance with (\ref{eq:zeta_mon}) and (\ref{eq:zeta_dec}), the number $|Z_{k,\ell}|$ of the states of automata with a memory depth $\ell$ refers to the number of substrings of length $\ell$. For instance, for a distributed automaton $P_{k,\ell}=(Z_{k,\ell},V_k,\Delta_k,Z_{0,k})$: $|Z_{k,\ell}|\leq \sum_{i=0}^\ell |V_k|^i$.
For comparison purposes, here, the maximal number of possible states in the automata implementation is considered. That is, the numbers $\sum_{k=1}^p\sum_{i=0}^\ell |V_k|^i$ (decentralized implementation) and $\sum_{i=0}^\ell |W|^i$ (monolithic implementation) need to be compared. As, by definition, $|V_k|<|W|$, it is obvious that for a sufficiently large $\ell$, the decentralized setup requires less memory space than the monolithic implementation. 

The memory requirements depend also significantly on the amount of the stored data $\chi_k(\zeta)$. For instance, for a finite state machine, a finite number of sets has to be stored. 
Similarly, for switched linear systems, the estimation sets are polytopes, which, again, can be stored as a finite number of sets representing the polytope vertices. 
%
%
%
%
Due to the exponential growth of the number of states with increasing $\ell$, the memory requirements can be significantly reduced by the decentralized approach at the price of additional online effort for the computation of set intersections. For certain classes of hybrid systems this operation can be implemented efficiently (such as polytope intersection in the case of switched linear systems). For other classes, the estimation sets can be offline overapproximated, e.g. by polyhedral methods, in order to obtain a computationally efficient online intersection operation. In general, a significant reduction in the offline computational time and required memory space is expected, as exemplified in the following section. 

%
%




\begin{figure}[!h]
\vspace{-10pt}
   \centering
   \includegraphics[width=0.4\textwidth]{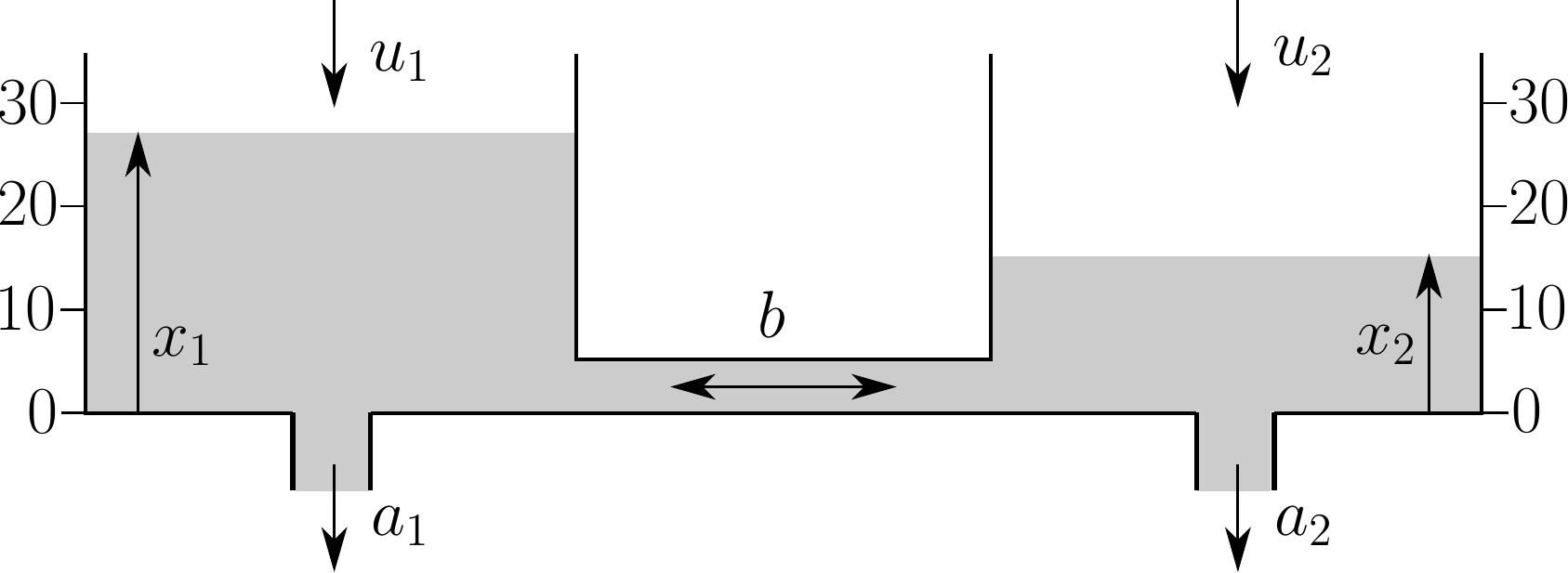}
\vspace{-10pt}
   \caption{Two-tank system.}
   \label{fig:tanks}
\vspace{-12pt}
\end{figure}

\subsection{Numerical example}
\label{sec:example}

%

The plant consists of two tanks with corresponding water levels $x_1$ and $x_2$, inflows $u_1$ and $u_2$, and outflow parameters $a_1$ and $a_2$, that are connected by a pipe with a flow constant $b$ (Fig.\,\ref{fig:tanks}). The dynamics can be described by the difference equations:
\begin{align}\label{eq:twotank}
	\hspm x(t\hspm+\hspm 1)&=\left[
\begin{array}{cc}
	\hspm 1\hspm-\hspm a_1\hspm-\hspm b\hspm\hspm & b\hspm\\
	\hspm b\hspm\hspm & 1\hspm-\hspm a_2\hspm-\hspm b\hspm
\end{array}\right]x(t)+ \left[
\begin{array}{cc}
	\hspm 1\hspm &\hspm\hspm 0\hspm \\
	\hspm 0\hspm & \hspm\hspm 1\hspm
\end{array}\right]u(t), 
\end{align}
where $t\in\mathbb{N}_0$, $x=[x_1,x_2]^T$, $u=[u_1,u_2]^T$, $y=[y_1,y_2]^T$, and $y(t)=x(t)$. In the following, the simulations will be performed for a set of input symbols $U_j\hspm=\hspm\{\mu_{j}^1, \hspm\mu_{j}^2, \hspm\mu_{j}^3\}$, where $\mu_{j}^1$, $\mu_{j}^2$, $\mu_{j}^3$ represent the inputs $u_{j}\hspm=\hspm 1$, $u_{j}\hspm=\hspm7$, $u_j\hspm=\hspm14$, respectively, with $j\in\{1,2\}$. Analogously, $Y_j\hspm=\hspm\{\nu_{j}^1, \hspm\nu_{j}^2, \hspm\nu_{j}^3\}$, where $\nu_{j}^1$, $\nu_{j}^2$, $\nu_{j}^3$ encode the measurements $y_{j}\hspm\in\hspm [0,10)$, $y_{j}\hspm\in\hspm [10,20)$, $y_{j}\hspm\in\hspm [20,30]$, respectively. Furthermore $a_1\hspm = \hspm a_2\hspm =\hspm 0.35$ and $b\hspm =\hspm 0.25$.

Note that this system assumes an infinite-dimensional I/S/O state machine representation $P=(X,W,\Delta,X)$, where the external signal space is given by $W=U_1\times U_2\times Y_1\times Y_2$. Hence, in accordance with \emph{Example~1}, only the output space is to be aggregated. For a decentralized setting including two state-machines, the resulting signal spaces of $P_{1,\ell}$ and $P_{2,\ell}$ are given by $V_1=U_1\times U_2  \times \mathcal{A}_1(Y_1\times Y_2)$ and $V_2=U_1\times U_2 \times \mathcal{A}_2(Y_1\times Y_2)$, respectively. Thereby, we adopt the mappings $\mathcal{A}_1$ and $\mathcal{A}_2$ according to $\{(\nu_{1}^i,\%)\}\mapsto\theta_{1}^i$ and $\{(\%,\nu_{2}^i)\}\mapsto\theta_{2}^i$, $i\in\{1,2,3\}$. As a result, $|V_1|+|V_2|=54$, while $|W|=81$, indicating a reduction in the computational complexity (see below).

The $\ell$-complete approximations of monolithic and distributed estimators with $\ell=2$ are simulated for $0\leq t\leq 7$ with the input sequence $u(t)$ as depicted in Fig.\,\ref{fig:sim} and $x(0)=[0, 0]^T$. Obviously, the reachable sets for an arbitrary sequence of input and measurement symbol pairs is represented by polytopes with a finite number of edges. As  system (\ref{eq:twotank}) fulfills the conditions of Corollary~\ref{cor:main_iso}, the decentralized estimation scheme must provide the exact outcome of the monolithic estimator. For instance, following \eqref{eq:zeta_mon} and \eqref{eq:zeta_dec}, at the time instant  $t=2$, the monolithic and distributed estimators reach the states
\begin{align*}
\zeta(2)&=\langle(\mu_{1}^2,\mu_{2}^2,\nu_{1}^1,\nu_{2}^1),(\mu_{1}^2,\mu_{2}^2,\nu_{1}^1,\nu_{2}^1),(\mu_{1}^2,\mu_{2}^2,\nu_{1}^2,\nu_{2}^2)\rangle, \\
\zeta_1(2)&=\langle(\mu_{1}^2,\mu_{2}^2,\theta_{1}^1),(\mu_{1}^1,\mu_{2}^2,\theta_{1}^1),(\mu_{1}^1,\mu_{2}^2,\theta_{1}^2)\rangle, \\
\zeta_2(2)&=\langle(\mu_{1}^2,\mu_{2}^2,\theta_{2}^1),(\mu_{1}^2,\mu_{2}^2,\theta_{2}^1),(\mu_{1}^2,\mu_{2}^2,\theta_{2}^2)\rangle,
\end{align*}
respectively, and the intersection of $\chi_1(\zeta_1(2))$ (the polygon with dashed edges) and $\chi_2(\zeta_2(2))$ (the polygon with dotted edges), is indeed equal to the set provided by the monolithic estimator (the gray area indicated by $\chi(\zeta(2))$), see Fig.\,\ref{fig:sim}.
%
%
Note that the true state of the system denoted by small filled circles is always contained in the corresponding state set.

Table~\ref{tab:twotankl} indicates that complexity of the estimators is strongly related to the number of states $|Z_\ell|$ of the $\ell$-complete automata used for realization of the approximations and the overall number of vertices of the associated sets $\chi(\zeta)$ given by $n_{\chi}$. Observe that, compared to the monolithic estimator, already for $\ell=2$, the decentralized setup with two estimators requires a lower offline computational cost $|Z_\ell|$, and amount of memory $n_\chi$.


\vspace{-5pt}
\begin{table}[!h]
	\centering
		\begin{tabular}{r|r@{  }r|r@{  }r|r@{  }r}
		\toprule
			$\ell$ & \multicolumn{2}{|c|}{$P_{\ell}$} & \multicolumn{2}{|c|}{$P_{1,\ell}$} & \multicolumn{2}{c}{$P_{1,\ell} \text{ \& } P_{2,\ell}$} \\
			& $|Z_\ell|$ & $n_{\chi}$ & $|Z_\ell|$ & $ n_{\chi}$  & $\Sigma |Z_\ell|$ & $ \Sigma n_{\chi}$ \\
			\midrule
			2& 172 & 993 & 54 & 307 & 108 & 617\\
			3& 2260 &17743 &703 &5509 &1406 & 11065\\
			\bottomrule
		\end{tabular}
	\caption{Complexity analysis.}
	\label{tab:twotankl}
\end{table}

\vspace{-23pt}
\section{Conclusions}
A general decentralized framework for set-valued state estimation and prediction for hybrid state machines has been discussed in this article. The outcome of the decentralized scheme is computed as the intersection of the sets provided by  the individual distributed state machines. The distributed state machines are defined as abstractions of the monolithic one, as specified by an appropriate decomposition of the original external signal space. In general, we show that decentralized schemes provide outer approximates of the estimates and predictions obtained by the monolithic computation.
Based on the concept of non-deterministic chains, we develop a simple decomposition algorithm for the external signal space, invariably leading to exact set-valued decentralized state estimation and prediction.
Due to the smaller cardinality of the constructed external signal spaces corresponding to the distributed state machines,
a significant reduction in the overall space/time computational complexity may result. This is illustrated here by applying $\ell$-complete approximation. Moreover, advantages in terms of robustness and reliability are gained. Indeed, in order to prevail the failure of a single processor, the decentralized scheme could be extended by a redundant state machine and a failure detection for switching off the malfunctioning processor. The proposed decentralized framework can be applied in different contexts, including data fusion, fault tolerant control, etc. Optimal signal space decomposition leading to a minimal computational complexity is a problem of interest for the future work. A further development of the algorithm based on non-deterministic chains can be found  in \cite{bajc_rom:2011}.


\begin{figure}[t]
\centering
\includegraphics[width=8.cm]{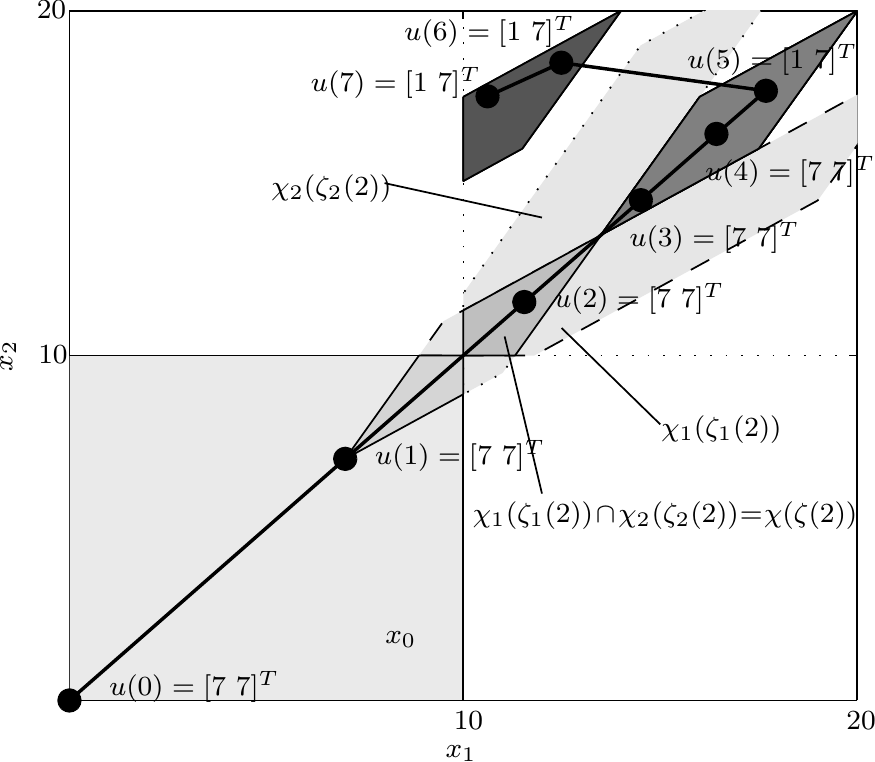}
\vspace{-10pt}
\caption{Set-valued state estimation using $\ell$-complete approximation ($\ell=2$).}
\label{fig:sim}
\vspace{-15pt}
\end{figure}


\bibliographystyle{IEEEtran}
\bibliography{bib/systol2010_literature}

\end{document}